\documentclass[10pt, twocolumn, letterpaper]{article}
\usepackage[margin=0.75in, columnsep=0.2in]{geometry}
\usepackage{graphicx}
\usepackage{amsmath, amssymb, amsfonts, amsthm}
\usepackage{booktabs}
\usepackage{authblk}
\usepackage{microtype}
\usepackage[numbers,sort&compress]{natbib}
\usepackage[colorlinks=true, linkcolor=blue, citecolor=blue, urlcolor=blue]{hyperref}

\newtheorem{proposition}{Proposition}
\newtheorem{lemma}{Lemma}
\newtheorem{corollary}{Corollary}
\newcommand{\Dg}{\mathcal{D}}
\newcommand{\GD}{\mathrm{GD}}
\newcommand{\ex}{\mathrm{ex}}
\newcommand{\loc}{\mathrm{loc}}

\title{Splicing an exact local self-energy into $GW$:\\
\large tests on small Hubbard clusters}
\author[1]{Michael O. Atambo\thanks{michael.atambo@tukenya.ac.ke}}
\affil[1]{Department of Physics, Earth and Environmental Science,
Technical University of Kenya, Nairobi, Kenya}
\date{\today}

\begin{document}

\twocolumn[{%
\maketitle
\begin{abstract}
$GW$+DMFT-type schemes replace the local part of a $GW$ self-energy by a separately computed local self-energy. A causal double counting removes a non-causal term from the hybridization (Lee and Haule, 2017), but the resulting total self-energy has been flagged as not guaranteed to be causal (Zhu and Chan, 2021). From the Herglotz representation, the spliced self-energy is causal if and only if its spectral measure is positive semi-definite. For translation-invariant single-band systems this becomes a momentum-resolved bound: the local weight may fall short of the local $GW$ weight by at most the momentum minimum of the $GW$ weight. On finite clusters, a mismatch between $GW$ and exact poles generically produces indefinite residues. On the Matsubara axis, a negative damping eigenvalue at one frequency certifies non-causality; passing does not establish causality. We examine these diagnostics in an idealized setting, not self-consistent $GW$+DMFT: one-shot $GW$ on the exact Green's function of small Hubbard clusters with the exact local self-energy inserted. At all 528 parameter sets, our computed $GW$ and exact self-energies pass both tests, as expected. The splice fails the pointwise test at 54 points, none at half filling, all at the lowest Matsubara frequency, and already at $U/t=0.5$; there the local $GW$ damping is 1.3-9 times the exact one. A block Pick test on eight frequencies, run with a second imaginary-time code, flags 481 of the 528 sets; for the other 47 it is inconclusive. Moment checks miss all of these, because the spliced high-frequency moment is exact. The Dyson Green's function is then non-causal too, yet passes the pointwise test everywhere tested. This version replaces a withdrawn one whose reported causality violations of $GW$ itself came from a sign error in the polarization; with the sign corrected, the computed $GW$ shows the expected positivity.
\end{abstract}
\vspace{2em}
}]

% ================================================================
\section{Introduction}
\label{sec:intro}

The $GW$ approximation~\cite{hedin65, strinati80, strinati82, hybertsen86, strinati88, aryasetiawan98, onida02, reining18, golze19} and dynamical mean-field theory (DMFT)~\cite{metzner89, georges92, georges96, kotliar06, held07} describe complementary physics: $GW$ captures nonlocal, dynamically screened exchange~\cite{strinati82, strinati84}, while DMFT treats local correlations nonperturbatively. Combining them, with dynamically screened local interactions~\cite{aryasetiawan04}, as in $GW$+DMFT, $GW$+EDMFT and related schemes~\cite{sun02, biermann03, sun04, ayral12, ayral13, tomczak12, biermann14, huang14, vanroekeghem14, tomczak15, choi16, boehnke16, nilsson17, golez17, choi19, zhu21}, self-energy embedding theory~\cite{zgid17, lan17, iskakov20}, and vertex-based extensions of DMFT~\cite{toschi07, katanin09, rubtsov08, rubtsov12, ayral15, stepanov19, rohringer18}, has become a central route to the electronic structure of correlated materials. Many of these constructions share one step: a nonlocal self-energy is built diagrammatically, and its local part is replaced by a local self-energy obtained elsewhere.

Causality is the natural question to ask of such a construction. For the exact self-energy $\Sigma$, the function $-\Sigma$ is a matrix-valued Herglotz function~\cite{luttinger61, gesztesy00}. Diagrammatic approximations that can be written as ``perfect squares'' on the Keldysh contour, $GW$ among them, inherit this property~\cite{stefanucci14, uimonen15}. This is distinct from being conserving ($\Phi$-derivable)~\cite{baym61, baym62, almbladh99}: approximations outside the perfect-square class, such as $GW$ with first-order vertex corrections, can give negative spectral functions~\cite{stefanucci14}. A self-energy assembled from pieces of different approximations carries no such guarantee. Non-causality has so far been reported mainly in \emph{auxiliary} quantities. Lee and Haule found that $GW$+DMFT for the H$_2$ molecule produces a non-causal hybridization in the strongly correlated regime~\cite{lee17, leethesis17}. They traced it to a mismatch between the local projection of the $GW$ self-energy and the double-counting term, and removed it with a ``causal double counting'' that subtracts exactly the local projection of $GW$. Non-causal impurity interactions $\mathcal{U}(\omega)$ are well known in $GW$+EDMFT~\cite{nilsson17, chen22}. Backes, Sim and Biermann showed that the Dyson-equation construction of the bath is non-causal in general and replaced it by a cavity construction that is causal by design~\cite{backes22}. Causality problems of periodized and interpolated cluster self-energies are an older theme of cluster DMFT~\cite{hettler98, lichtenstein00, kotliar01, hettler00, biroli02, biroli04, maier05, stanescu06, pelz26}. In these works, the physical lattice Green's function and self-energy were found or assumed to remain causal~\cite{nilsson17, chen22, lan17}.

The total self-energy itself has received much less attention. Zhu and Chan, who subtract the local $G_0W_0$ self-energy and insert an impurity self-energy, note explicitly that, with one-shot $G_0W_0$, the result is not guaranteed to be strictly causal, use bath optimization and a small broadening to avoid non-causal behavior, and report that no low-energy negative spectral weight was observed in their calculations~\cite{zhu21}. Lee and Haule's argument for their causal double counting concerns the hybridization in the limit of a large impurity self-energy~\cite{lee17, leethesis17}. With that double counting, their total self-energy has the same algebraic structure as the splice studied here, although their calculation is otherwise different (Sec.~\ref{sec:method}).

In this paper we examine the question for the spliced self-energy itself, in a controlled setting. We work on small Hubbard clusters, where every ingredient is available exactly. We insert the \emph{exact} local self-energy, so no impurity-solver or bath error can be blamed, and evaluate $GW$ on the exact Green's function. The results are:
(i)~a spectral-positivity condition for the spliced self-energy, which follows directly from the Herglotz representation, together with three consequences: a momentum-resolved bound for translation-invariant systems, a generic failure mechanism on finite clusters, and a Matsubara-axis (pointwise) form that is sufficient in general and necessary and sufficient at each frequency for translation-invariant clusters and single-band lattices (Sec.~\ref{sec:criterion});
(ii)~a symmetry argument showing that the half-filled particle-hole symmetric dimer passes this test at every Matsubara frequency (Sec.~\ref{sec:criterion});
(iii)~a short argument that the spliced high-frequency moment is always the exact, positive semi-definite local moment, so moment checks are blind to the failure (Sec.~\ref{sec:tail});
(iv)~an exact-diagonalization survey of 528 parameter sets in which $GW$ always passes, while the splice fails the pointwise test at 54 points and an eight-frequency block Pick test at 481, the latter evaluated with a second implementation (Sec.~\ref{sec:results}).
The setting is deliberately idealized: one-shot $GW$ on the exact cluster Green's function with the exact local self-energy inserted, not a self-consistent $GW$+DMFT calculation.
Section~\ref{sec:discussion} discusses the consequences for production $GW$+DMFT calculations.

% ================================================================
\section{Model, self-energies and the causality test}
\label{sec:method}

\subsection{Clusters and exact solution}
We consider the Hubbard model~\cite{hubbard63} on finite clusters,
\begin{equation}
H = -t\!\!\sum_{\langle ij\rangle\sigma}\!(c^\dagger_{i\sigma}c_{j\sigma}+\mathrm{h.c.})
+ U\sum_i d_{i\uparrow}d_{i\downarrow} - \mu N,
\label{eq:ham}
\end{equation}
with $d_{i\sigma}=n_{i\sigma}-\tfrac12$, so that $\mu=0$ is half filling on bipartite clusters. We study the dimer, the open three-site chain, the three-site ring and the four-site ring at $t=1$. We diagonalize $H$ in the full Fock space (16, 64, 64 and 256 states) and evaluate the exact Green's function $G$ from its Lehmann representation. The exact self-energy is $\Sigma^{\ex}(i\omega_n)=i\omega_n I - h - G^{-1}(i\omega_n)$, with $h$ the one-body matrix including $-\mu$. The Hubbard dimer is a standard test bed for exactly this kind of question~\cite{carrascal15, romaniello09, romaniello12, lee17}.

\subsection{$GW$ on the exact Green's function}
We evaluate one-shot $GW$ on the exact $G$. The polarization (spin factor 2), screened interaction and self-energy are
\begin{align}
P_{ij}(i\nu_m) &= \frac{2}{\beta}\sum_n G_{ij}(i\omega_n+i\nu_m)\,G_{ji}(i\omega_n), \label{eq:P}\\
W &= (1-UP)^{-1}U, \label{eq:W}\\
\Sigma^{GW}_{ij}(i\omega_n) &= \Sigma^{H}\delta_{ij} \nonumber\\
&\quad -\frac{1}{\beta}\sum_m G_{ij}(i\omega_n+i\nu_m)\,[W-U]_{ij}(i\nu_m). \label{eq:sgw}
\end{align}
The product in Eq.~\eqref{eq:sgw} is element-wise in the site basis. The sign of Eq.~\eqref{eq:P} is fixed by requiring $P$ to reproduce the exact density response at $U=0$ (agreement to $2.5\times10^{-4}$ on the ring), so that $P(0)\preceq0$ and $W$ is screened. Static terms such as the Hartree shift $\Sigma^H$ are Hermitian and play no role in what follows. Evaluating $GW$ on the exact $G$ isolates the effect of the splice from that of self-consistency, whether full~\cite{holm98}, quasiparticle~\cite{vanschilfgaarde06} or vertex-corrected~\cite{shishkin07, romaniello09, romaniello12}. Numerical details, including an analytic correction of the truncated Matsubara sum in Eq.~\eqref{eq:P}, are given in App.~\ref{app:numerics}.

\subsection{The spliced self-energy}
Let $\Dg$ denote the projection onto the site-diagonal. The spliced self-energy is
\begin{equation}
\Sigma^{\GD} = \Sigma^{GW} + \Dg\big[\Sigma^{\loc} - \Sigma^{GW}\big],
\label{eq:splice}
\end{equation}
with $\Sigma^{\loc}=\Dg\,\Sigma^{\ex}$ the exact local self-energy. Equation~\eqref{eq:splice} is the cluster analogue of
$\Sigma(\mathbf k)=\Sigma^{GW}(\mathbf k)-\Sigma^{GW}_{\loc}+\Sigma_{\mathrm{imp}}$,
i.e.\ $GW$+DMFT with the causal double counting $\Sigma_{\mathrm{DC}}=\hat P\Sigma^{GW}$ of Ref.~\cite{lee17} and the local $G_0W_0$ subtraction of Ref.~\cite{zhu21}. The comparison is at the level of the subtraction only: unlike those calculations, we insert the exact local self-energy instead of an impurity self-energy and evaluate $GW$ once, on the exact Green's function.

\subsection{Causality on the Matsubara axis}
A causal self-energy has the representation
$\Sigma(z)=\Sigma_\infty+\int d\varepsilon\,B(\varepsilon)/(z-\varepsilon)$ with Hermitian $\Sigma_\infty$ and $B(\varepsilon)\succeq0$~\cite{luttinger61, strinati88, gesztesy00}. Define the damping matrix at $\omega_n>0$~\cite{stefanucci13},
\begin{equation}
\Gamma(i\omega_n) \equiv -\frac{\Sigma(i\omega_n)-\Sigma(i\omega_n)^\dagger}{2i}
= \omega_n\!\int\! d\varepsilon\,\frac{B(\varepsilon)}{\omega_n^2+\varepsilon^2}.
\label{eq:gamma}
\end{equation}
For a real Hamiltonian $\Sigma_{ij}=\Sigma_{ji}$, so $\Gamma=-\mathrm{Im}\,\Sigma$ element-wise, and $\Sigma(-i\omega_n)=\Sigma(i\omega_n)^\dagger$, so positive frequencies suffice. Following the many-body literature~\cite{lee17, backes22, nilsson17}, ``causal'' here means that $-\Sigma$ has the matrix Herglotz property: $\Sigma$ is analytic in the upper half-plane and has a positive semi-definite spectral function. The spliced self-energy studied below is always analytic in the upper half-plane and obeys the Kramers-Kronig relations, since each of its ingredients does; what can fail is the positivity of its spectral weight. In the terminology of electromagnetism and network theory, where analyticity and positivity are called causality and passivity~\cite{srivastava21}, such a self-energy is causal but not passive.

\begin{lemma}
\label{lem:test}
If $\Sigma$ is causal, then $\Gamma(i\omega_n)\succeq0$ for every $\omega_n>0$. Hence a negative eigenvalue of $\Gamma$ at a single Matsubara frequency proves that no causal self-energy takes the given Matsubara values.
\end{lemma}

The pointwise test of Lemma~\ref{lem:test} is necessary but not sufficient: a self-energy that passes it at every sampled frequency need not be causal. For the finite interpolation problem in the matrix-Herglotz class, the Pick criterion used in Nevanlinna and Carath\'eodory continuation~\cite{fei21, fei21b} is necessary and sufficient for the sampled values of $F=-\Sigma$ to admit a Herglotz interpolant. A failure therefore excludes every causal self-energy; a pass does not show that the interpolant has the asymptotics of a self-energy, nor that the underlying self-energy is causal. The criterion requires the block matrix
\begin{equation}
\mathcal P_{ab} = \frac{\Sigma(i\omega_b)^\dagger-\Sigma(i\omega_a)}{i(\omega_a+\omega_b)},\qquad a,b=1,\dots,N_P,
\label{eq:pick}
\end{equation}
to be positive semi-definite; its diagonal blocks are $\Gamma(i\omega_a)/\omega_a$, so the pointwise test is the block-diagonal part of the Pick test. We use the pointwise test to \emph{prove} non-causality, and the Pick test to probe the points that pass it (Sec.~\ref{sec:results}). No analytic continuation~\cite{vidberg77, jarrell96} is needed for either. We quantify a violation by
\begin{equation}
v = \max\Big(0,\;-\min_{n\le N_\omega}\lambda_{\min}\big[\Gamma(i\omega_n)\big]\Big),
\label{eq:v}
\end{equation}
with $N_\omega=50$ positive frequencies.

% ================================================================
\section{Causality conditions for the splice}
\label{sec:criterion}

\subsection{Spectral-positivity condition}
Both ingredients of the splice are causal: the exact self-energy by construction, and $GW$ because it is a perfect-square diagram~\cite{stefanucci14} (App.~\ref{app:gw}). Each therefore has a spectral representation with a positive semi-definite spectral function, $B^{\ex}(\varepsilon)\succeq0$ and $B^{GW}(\varepsilon)\succeq0$. Because Eq.~\eqref{eq:splice} is linear, the splice has the same representation with
\begin{equation}
B^{\GD}(\varepsilon) = B^{GW}(\varepsilon) + \Dg\big[B^{\loc}(\varepsilon)-B^{GW}(\varepsilon)\big],
\label{eq:BGD}
\end{equation}
a Hermitian but in general indefinite matrix-valued measure. The following statement is a direct consequence of the Herglotz representation and the linearity of the splice; we state it to fix notation for its consequences.

\begin{proposition}[spectral positivity]
\label{thm:global}
$\Sigma^{\GD}$ is causal if and only if $B^{\GD}$ is a positive semi-definite matrix-valued measure, i.e.\ $\int_{\Omega}v^\dagger dB^{\GD}(\varepsilon)\,v\ge0$ for every vector $v$ and every interval $\Omega$ of the real axis. For a discrete spectrum, $B^{\GD}=\sum_pR_p\,\delta(\varepsilon-\varepsilon_p)$, this means $R_p\succeq0$ for every $p$. In particular:
(a)~it is causal if $B^{\loc}_{ii}(\varepsilon)\ge B^{GW}_{ii}(\varepsilon)$ for all $i$ and $\varepsilon$;
(b)~for a translation-invariant cluster or a single-band lattice it is causal if and only if
\begin{equation}
B^{\loc}(\varepsilon) \;\ge\; B^{GW}_{\loc}(\varepsilon)-\min_{\mathbf k}B^{GW}(\mathbf k,\varepsilon)\quad\text{for all }\varepsilon .
\label{eq:global}
\end{equation}
\end{proposition}
\begin{proof}
The spectral measure of a function with this representation is fixed by Stieltjes inversion, $B=-\pi^{-1}\lim_{\eta\to0^+}\mathrm{Im}\,\Sigma(\varepsilon+i\eta)$, and the negative of the function is a matrix Herglotz function if and only if that measure is positive semi-definite~\cite{gesztesy00}. (a)~$B^{\GD}$ is then a positive semi-definite matrix plus a non-negative diagonal one. (b)~All diagonal elements are equal, so Eq.~\eqref{eq:BGD} shifts every eigenvalue $B^{GW}(\mathbf k,\varepsilon)$ by $B^{\loc}(\varepsilon)-B^{GW}_{\loc}(\varepsilon)$.
\end{proof}

\begin{corollary}[finite clusters]
\label{cor:finite}
On a finite cluster all spectral measures are discrete, $B(\varepsilon)=\sum_pR_p\,\delta(\varepsilon-\varepsilon_p)$. If $\varepsilon_p$ is a pole of $\Sigma^{GW}_{ij}$ ($i\ne j$) with non-zero residue but not a pole of $\Sigma^{\loc}_{ii}$ or $\Sigma^{\loc}_{jj}$, the residue $R^{\GD}_p$ restricted to sites $i,j$ has the form $\bigl(\begin{smallmatrix}0&r\\ r^*&0\end{smallmatrix}\bigr)$ with $r\ne0$, whose eigenvalues are $\pm|r|$. It is therefore indefinite, and $\Sigma^{\GD}$ is not causal.
\end{corollary}

The poles of $\Sigma^{GW}$ lie at the single-particle excitation energies shifted by the RPA collective energies, while the poles of the exact self-energy are set by the zeros of $\det G$. They do not coincide in general, so on a finite cluster the splice is generically non-causal. The argument needs both a pole mismatch and a non-zero off-diagonal residue; it identifies a generic mechanism, not a proof that every finite-cluster splice fails at every parameter value. On a lattice with continuous spectra, Eq.~\eqref{eq:global} is instead a genuine inequality that may or may not hold. The off-diagonal $GW$ spectral weight integrates to zero, $\int d\varepsilon\,B^{GW}_{ij}(\varepsilon)=0$ for $i\neq j$ (Sec.~\ref{sec:tail}), so it necessarily takes both signs. Without an equally large local weight at the same energies, it makes $B^{\GD}$ indefinite.

\subsection{What the Matsubara axis shows}
Calculations are usually done on the Matsubara axis, where only a smoothed version of Proposition~\ref{thm:global} is visible. From Eq.~\eqref{eq:gamma},
\begin{equation}
\Gamma^{\GD}(i\omega_n) = \omega_n\!\int\! d\varepsilon\,\frac{B^{\GD}(\varepsilon)}{\omega_n^2+\varepsilon^2}
= \Gamma^{GW} + \Dg\big[\Gamma^{\loc}-\Gamma^{GW}\big].
\label{eq:gammaGD}
\end{equation}
The positive kernel averages $B^{\GD}$ over a window of width $\sim\omega_n$, so $\Gamma^{\GD}(i\omega_n)$ can be positive semi-definite even when $B^{\GD}$ is not. The pointwise conditions below are therefore necessary consequences of causality, not equivalents. Their value is that a failure at any single frequency is a proof of non-causality that needs no analytic continuation.

\begin{proposition}[sufficient condition, pointwise]
\label{prop:suff}
If $\Gamma^{\loc}_{ii}(i\omega_n)\ge\Gamma^{GW}_{ii}(i\omega_n)$ for all sites $i$ at a given $\omega_n$, then $\Gamma^{\GD}(i\omega_n)\succeq0$.
\end{proposition}
\begin{proof}
$\Gamma^{\GD}$ is then a positive semi-definite matrix plus a non-negative diagonal matrix.
\end{proof}

The splice is therefore safe whenever the inserted local self-energy damps at least as strongly as $GW$ does locally. A violation needs the opposite: $GW$ must \emph{overdamp} locally, so that removing its local part removes more damping than its nonlocal part can tolerate.

\begin{proposition}[translation-invariant case, pointwise]
\label{prop:iff}
For a translation-invariant cluster or a single-band lattice, let $\Gamma^{GW}(\mathbf k,i\omega_n)\ge0$ be the eigenvalues of $\Gamma^{GW}$, $\gamma^{GW}=N^{-1}\sum_{\mathbf k}\Gamma^{GW}(\mathbf k)$ its local value and $\gamma^{\loc}$ the inserted local damping. Then $\Gamma^{\GD}(i\omega_n)\succeq0$ if and only if
\begin{equation}
\gamma^{GW}(i\omega_n)-\gamma^{\loc}(i\omega_n) \;\le\; \min_{\mathbf k}\Gamma^{GW}(\mathbf k,i\omega_n).
\label{eq:iff}
\end{equation}
\end{proposition}
\begin{proof}
All diagonal elements are equal, so $\Dg[\Gamma^{\loc}-\Gamma^{GW}]=(\gamma^{\loc}-\gamma^{GW})I$. This shifts every eigenvalue $\Gamma^{GW}(\mathbf k)$ by the same amount, and the smallest one must stay non-negative.
\end{proof}

In words, at a given frequency the inserted local damping may undercut that of $GW$ only by the momentum minimum of the $GW$ damping. Propositions~\ref{prop:suff} and~\ref{prop:iff} are statements about $\Gamma^{\GD}$ at one Matsubara frequency. A failure at any frequency proves non-causality (Lemma~\ref{lem:test}); satisfying them at every Matsubara frequency is necessary for causality but not sufficient.

\begin{corollary}[dimer]
\label{cor:dimer}
For the dimer, $\Gamma^{GW}(k)=\gamma^{GW}\pm\Gamma^{GW}_{12}$, and $\Gamma^{\GD}(i\omega_n)\succeq0$ if and only if $\gamma^{\loc}(i\omega_n)\ge|\Gamma^{GW}_{12}(i\omega_n)|$.
\end{corollary}

\begin{proposition}[particle-hole symmetry, pointwise]
\label{prop:ph}
On a bipartite cluster at $\mu=0$, $\Sigma^{\ex}_{ij}(i\omega_n)$ and $\Sigma^{GW}_{ij}(i\omega_n)$ are real for $i,j$ on opposite sublattices, so the corresponding damping elements vanish. In particular, for the half-filled dimer $\Gamma^{\GD}(i\omega_n)\succeq0$ at every Matsubara frequency, for all $U$ and $\beta$.
\end{proposition}
\begin{proof}
With $S=\mathrm{diag}[(-1)^i]$, the particle-hole transformation $c_{i\sigma}\to(-1)^ic^\dagger_{i\sigma}$ leaves $H(\mu=0)$ invariant and gives $G^*(i\omega_n)=-SG(i\omega_n)S$. Since $ShS=-h$, the same relation holds for $G^{-1}$ and hence for $\Sigma^{\ex}$. $P$ and $W$ built from $G$ are real, and $G_{ij}$ is real for opposite-sublattice pairs, so the same holds for $\Sigma^{GW}_{ij}$. For the dimer, $\Gamma^{\GD}$ is then diagonal and equals $\Gamma^{\loc}\succeq0$.
\end{proof}

Proposition~\ref{prop:ph} is a statement on the Matsubara axis only. Particle-hole symmetry makes $B^{GW}_{12}(\varepsilon)$ odd in $\varepsilon$, so its Matsubara projection vanishes; $B^{GW}_{12}$ itself does not. By Corollary~\ref{cor:finite} the half-filled dimer splice is therefore still non-causal, which the pointwise test cannot see but the Pick test does (Sec.~\ref{sec:results}). Same-sublattice elements (sites 0 and 2 of the open chain and of the four-site ring) are purely imaginary at $\mu=0$, so Proposition~\ref{prop:ph} gives no guarantee beyond the dimer. We checked both statements numerically: the inter-sublattice damping vanishes to $10^{-12}$ at $\mu=0$, and the identity behind Eq.~\eqref{eq:iff} holds to $4\times10^{-12}$ on the dimer and the three-site ring.

% ================================================================
\section{High-frequency moments cannot detect the failure}
\label{sec:tail}

Writing $\Sigma(z)=\Sigma_\infty+C/z+O(z^{-2})$, Eq.~\eqref{eq:gamma} gives $\omega_n\Gamma(i\omega_n)\to C$, so $C\succeq0$ is a necessary condition for causality. The exact moment follows from the spectral moments of $G$~\cite{harris67, potthoff97} (App.~\ref{app:moments}):
\begin{equation}
C^{\ex}_{ij} = U^2\,\langle n_{i\bar\sigma}\rangle\big(1-\langle n_{i\bar\sigma}\rangle\big)\,\delta_{ij}\;\succeq\;0 .
\label{eq:cex}
\end{equation}
It is exactly local and positive.

For $GW$, the leading term $G(z)=I/z+O(z^{-2})$ is local, and $[W-U](i\nu_m)=O(\nu_m^{-2})$. Inserted into the element-wise product of Eq.~\eqref{eq:sgw}, this makes the off-diagonal elements of $\Sigma^{GW}$ decay as $O(z^{-2})$. The $GW$ moment is therefore also diagonal,
\begin{equation}
C^{GW}_{ij} = -\delta_{ij}\,\frac{1}{\beta}\sum_m [W-U]_{ii}(i\nu_m)\;\ge 0 ,
\label{eq:cgw}
\end{equation}
which is positive because $W-U=U\chi^{\mathrm{RPA}}U$ with $\chi^{\mathrm{RPA}}(i\nu_m)\preceq0$. It differs from Eq.~\eqref{eq:cex} in value: it is an RPA fluctuation rather than the exact local variance.

\begin{proposition}[tail of the splice]
\label{prop:tail}
For any splice of the form~\eqref{eq:splice} whose local part has the exact moment, $C^{\GD}=C^{\ex}$. The spliced moment is diagonal and positive semi-definite, so the condition $C\succeq0$ is always satisfied.
\end{proposition}

The moment test is therefore structurally blind to splice failures. Since $C=\int d\varepsilon\,B(\varepsilon)$, the vanishing off-diagonal moment of $GW$ also means that $\int d\varepsilon\,B^{GW}_{ij}(\varepsilon)=0$ for $i\ne j$: the off-diagonal $GW$ spectral weight takes both signs, which is what makes $B^{\GD}$ in Eq.~\eqref{eq:BGD} indefinite wherever the local weight is too small. Numerically, $C^{GW}$ is positive at all 528 points (smallest eigenvalue 0.008), and its off-diagonal elements are below $2\times10^{-4}$. For every violating point, $C^{\GD}$ agrees with the exact local moment to $10^{-4}$.

% ================================================================
\section{Exact-diagonalization results}
\label{sec:results}

We first report the pointwise test, which proves non-causality where it fails, then the Pick test, which reveals how much more widespread non-causality is. We scanned $\beta t\in\{10,20\}$, $U/t\in\{0.5,1,2,3,4,6\}$ and $\mu/t\in\{-3,-2.5,\dots,3\}$ on all four clusters (624 points). We kept the 528 points with $0.1<\langle N\rangle<2N_s-0.1$, i.e.\ away from the empty and full band, where all self-energies vanish. Table~\ref{tab:summary} and Fig.~\ref{fig:maps} summarize the results.

\begin{table*}[t]
\caption{Summary of the scan. ``Pointwise'' counts points where $\Gamma^{\GD}(i\omega_n)$ has a negative eigenvalue at one of the $N_\omega=50$ lowest frequencies ($v>10^{-6}$); $v_{\max}$ is the largest such violation and ``ratio'' the range of $\gamma^{GW}/\gamma^{\loc}$ at $\omega_1$ over these points. ``Suff.'' counts points where Proposition~\ref{prop:suff} holds at all 50 frequencies. ``Pick'' counts points where the scaled block Pick matrix~\eqref{eq:pick} built from the lowest $N_P$ frequencies has an eigenvalue below $-10^{-9}$, computed with the second implementation of App.~\ref{app:validation}. The exact and $GW$ self-energies fail none of these tests at any point.}
\label{tab:summary}
\centering
\begin{tabular}{@{}lccccccc@{}}
\toprule
 & & \multicolumn{4}{c}{pointwise test} & \multicolumn{2}{c}{Pick test} \\
\cmidrule(lr){3-6}\cmidrule(l){7-8}
Cluster & points & failures & $v_{\max}$ & ratio & suff. & $N_P=2$ & $N_P=8$ \\
\midrule
dimer        & 118 & 14 & 0.068 & 1.63-3.24 & 11 & 42 & 118 \\
open 3-chain & 132 & 14 & 0.060 & 1.76-8.99 & 9  & 52 & 124 \\
3-site ring  & 132 & 6  & 0.024 & 1.34-2.73 & 8  & 46 & 118 \\
4-site ring  & 146 & 20 & 0.016 & 1.30-2.65 & 9  & 71 & 121 \\
\midrule
total        & 528 & 54 &       & 1.30-8.99 & 37 & 211 & 481 \\
\bottomrule
\end{tabular}
\end{table*}

\begin{figure*}[t]
\centering
\includegraphics[width=\textwidth]{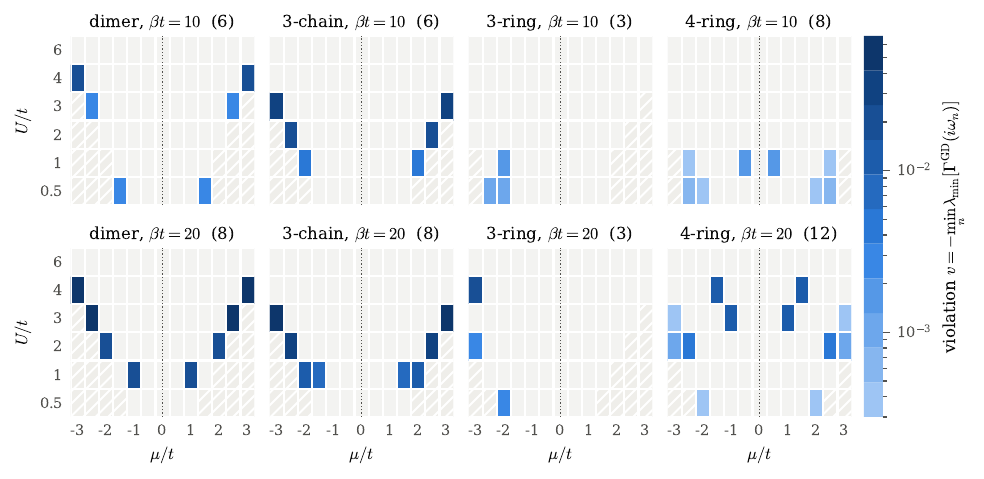}
\caption{Pointwise violation $v$ of the spliced self-energy, Eq.~\eqref{eq:v}, over the $(\mu,U)$ grid for all four clusters at $\beta t=10$ (top) and $20$ (bottom). Light cells pass the pointwise test (most of them fail the Pick test, Table~\ref{tab:summary}); hatched cells lie outside the filling window (empty or full band). The number of pointwise failures is given in each title. The dotted line marks $\mu=0$. $GW$ alone passes at every point.}
\label{fig:maps}
\end{figure*}

\paragraph{$GW$ passes the test.}
At all 528 points, $\Gamma^{GW}(i\omega_n)\succeq0$ at all 50 frequencies, and its tail is positive. $GW$ is causal by the perfect-square theorem~\cite{stefanucci14} (App.~\ref{app:gw}); the numerical result is a consistency check of the implementation, not an independent proof.

\paragraph{The splice fails it.}
At 54 points, $\Gamma^{\GD}$ has a negative eigenvalue, which by Lemma~\ref{lem:test} proves that $\Sigma^{\GD}$ is not causal: 14 on the dimer, 14 on the open chain, 6 on the three-site ring and 20 on the four-site ring. The largest violation is $v=0.068$ (dimer, $\beta t=20$, $U/t=4$, $\mu/t=\pm3$); the smallest is $1.6\times10^{-5}$. All values are stable to $1.1\times10^{-6}$ under a doubling of all Matsubara cutoffs, and no point changes status.

\paragraph{A block Pick test detects more.}
The block Pick matrix of Eq.~\eqref{eq:pick}, built from the lowest $N_P$ frequencies, confirms all 54 pointwise failures and detects many more (Table~\ref{tab:summary}): 211 of the 528 points fail with $N_P=2$ and 481 with $N_P=8$. These values come from the second implementation (App.~\ref{app:validation}), whose damping eigenvalues change by less than $2\times10^{-12}$ under quadrature refinement. The exact and $GW$ self-energies pass at every point; their smallest scaled Pick eigenvalues are $-4\times10^{-12}$ and $-1\times10^{-15}$, i.e.\ roundoff. We therefore count a failure only below $-10^{-9}$. The main implementation gives the same $N_P=2$ counts. The failures include 39 of the 48 half-filled points at $N_P=8$, among them all 12 half-filled dimer points, where the pointwise test passes by Proposition~\ref{prop:ph}. The increase of the counts with $N_P$ is consistent with the mechanism of Corollary~\ref{cor:finite}: more Matsubara frequencies resolve more of the indefinite real-axis weight. The corollary does not predict the counts themselves. At the 47 points that pass at $N_P=8$, the smallest scaled Pick eigenvalue lies between $-10^{-9}$ and roundoff, so these points are unresolved rather than shown to be causal. The counts are lower bounds on non-causality within this survey. They are properties of the chosen clusters, parameter grid, exact-$G$ construction and splice, not a failure rate for $GW$+DMFT calculations in general.

\paragraph{The Dyson Green's function.}
Since $\Sigma^{\GD}$ tends to a constant, $G^{\GD}=(z-h-\Sigma^{\GD})^{-1}$ has the leading behavior $I/z$. If $G^{\GD}$ were causal, $\Sigma^{\GD}=z-h-(G^{\GD})^{-1}$ would be causal as well~\cite{luttinger61, gesztesy00}. $G^{\GD}$ is therefore non-causal wherever $\Sigma^{\GD}$ is. On the Matsubara axis this is harder to see. $G^{\GD}$ passes the pointwise test at every point we checked: all 54 points where $\Sigma^{\GD}$ fails it, and all 382 points of the three small clusters (smallest eigenvalue $0.011$). The Pick test on $G^{\GD}$ fails at 135 points with $N_P=2$ and 454 with $N_P=8$. Matsubara data do not determine where on the real axis the non-causality of $G^{\GD}$ shows up, whether as negative spectral weight or as poles off the real axis. We therefore make no claim about specific real-frequency features of $G^{\GD}$.

\paragraph{Validation.}
Because Version~1 of this work was invalidated by a sign error, we validated the corrected implementation in three independent ways (App.~\ref{app:validation}).
(i)~At $U=0$ the polarization reproduces the exact density response to $1.1\times10^{-6}$; the Version~1 sign misses it by $2|\chi|$ (Table~\ref{tab:validation}).
(ii)~At all 528 points, $P(i\nu_m)\preceq0$ and the eigenvalues of $1-UP$ are $\ge1$, so $W$ is screened. With the Version~1 sign, $1+UP(0)$ is singular or negative at 410 points.
(iii)~A second implementation, sharing no code with the main one, reproduces the smallest eigenvalue of $\Gamma^{\GD}(i\omega_1)$ at nine representative points to $\le1.5\times10^{-7}$ (Table~\ref{tab:indep}). It works in imaginary time with its own exact diagonalization and no truncated Matsubara sums. The nine points include the smallest pointwise violation and the Version~1 headline case. The two kinds of check are distinct: the pointwise counts and violation values in Table~\ref{tab:summary} come from the main implementation and the Pick counts from the second one. The two implementations agree directly at the nine points of Table~\ref{tab:indep} and in the $N_P=2$ Pick counts; not every reported number was computed twice.

\paragraph{Where the pointwise failures occur.}
Three features stand out in Fig.~\ref{fig:maps}.
First, no cluster violates at $\mu=0$. For the dimer this is Proposition~\ref{prop:ph}; for the open chain and the four-site ring it is an observation.
Second, on the bipartite clusters violations come in $\mu\leftrightarrow-\mu$ pairs related by particle-hole symmetry. The non-bipartite three-site ring violates only on the hole-doped side ($\mu<0$).
Third, violations already appear at weak coupling, $U/t=0.5$, on three of the four clusters. They are not a strong-correlation effect.
Every violation occurs at $\omega_1$, nine of them also at $\omega_2$, and none at higher frequencies. The failure is a low-energy one, invisible in the tail (Sec.~\ref{sec:tail}). Pointwise failures mark the parameters where the indefinite part of $B^{\GD}$ is large enough at low energy to survive the Matsubara smoothing of Eq.~\eqref{eq:gammaGD}.

\paragraph{Mechanism.}
Figure~\ref{fig:mechanism} dissects the dimer at $\beta t=20$, $U/t=3$, $\mu/t=2.5$ ($\langle N\rangle=3.33$). At $\omega_1$, $GW$ damps the local channel 2.5 times more strongly than the exact solution ($\Gamma^{GW}_{11}=0.107$ versus $\Gamma^{\ex}_{11}=0.042$), and its nonlocal damping is $|\Gamma^{GW}_{12}|=0.085$. By Corollary~\ref{cor:dimer}, $\lambda_{\min}=\Gamma^{\ex}_{11}-|\Gamma^{GW}_{12}|=-0.043$, which is the violation measured. At higher frequencies the exact local damping catches up, and $\omega_n\lambda_{\min}$ approaches the exact moment $U^2n_\sigma(1-n_\sigma)$ from below, as Proposition~\ref{prop:tail} requires.

\begin{figure}[t]
\centering
\includegraphics[width=\columnwidth]{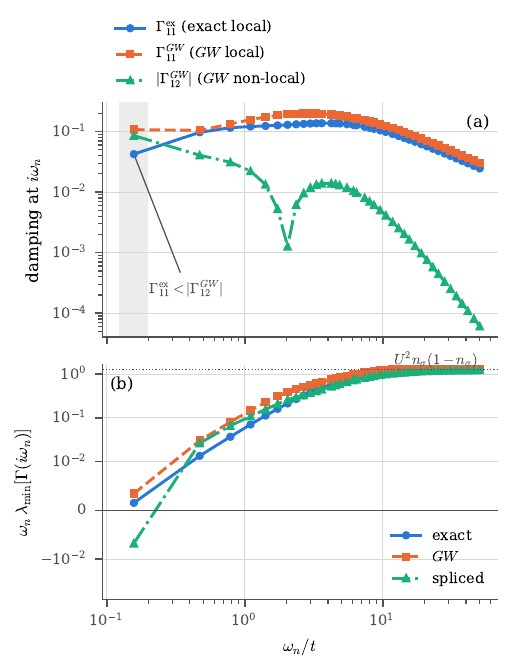}
\caption{Dimer at $\beta t=20$, $U/t=3$, $\mu/t=2.5$. (a)~Local exact, local $GW$ and nonlocal $GW$ damping. In the shaded region $\Gamma^{\ex}_{11}<|\Gamma^{GW}_{12}|$, and the splice violates causality (Corollary~\ref{cor:dimer}). (b)~$\omega_n\lambda_{\min}[\Gamma(i\omega_n)]$ for the exact, $GW$ and spliced self-energies (symmetric-log scale). The spliced curve is negative at $\omega_1$ yet tends to the exact moment (dotted line) at high frequency.}
\label{fig:mechanism}
\end{figure}

\paragraph{Test of the pointwise criterion.}
Figure~\ref{fig:criterion}(a) plots the violation against $\gamma^{GW}/\gamma^{\loc}$ at $\omega_1$. Every violating point has a ratio above 1, as Proposition~\ref{prop:suff} requires. Many points with ratio above 1 still pass the pointwise test: local overdamping is necessary for a pointwise violation, not sufficient. Figure~\ref{fig:criterion}(b) tests the necessary and sufficient condition~\eqref{eq:iff} on the translation-invariant clusters. Points that fail and pass the pointwise test are separated exactly by the diagonal. The 222 translation-invariant points with no local deficit all pass it.

\begin{figure*}[t]
\centering
\includegraphics[width=\textwidth]{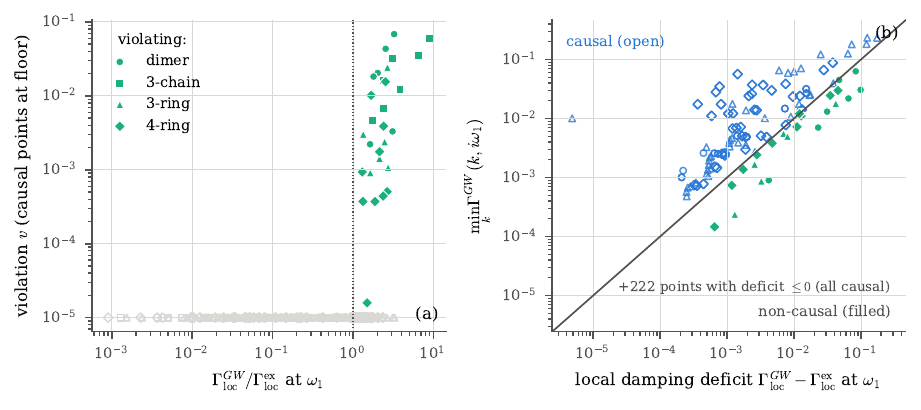}
\caption{(a)~Pointwise violation $v$ versus the ratio of $GW$ to exact local damping at $\omega_1$ for all 528 points; points that pass the pointwise test are drawn at the floor $10^{-5}$. All violations have ratio $>1$ (dotted line). (b)~Pointwise criterion, Eq.~\eqref{eq:iff}, for the dimer and the two rings at $\omega_1$: local damping deficit versus the momentum minimum of the $GW$ damping. Filled symbols fail the pointwise test and lie below the diagonal; open symbols pass it.}
\label{fig:criterion}
\end{figure*}

% ================================================================
\section{Discussion}
\label{sec:discussion}

\paragraph{Relation to previous work.}
Lee and Haule~\cite{lee17, leethesis17} showed that the hybridization of $GW$+DMFT acquires a non-causal term $\propto(\hat P\Sigma^{GW}-\Sigma_{\mathrm{DC}})$. They removed it by setting $\Sigma_{\mathrm{DC}}=\hat P\Sigma^{GW}$, and checked causality of $G_{\mathrm{imp}}$, $\Sigma_{\mathrm{imp}}$ and $\Delta$. The total self-energy with that double counting is Eq.~\eqref{eq:splice}, and our results show that it is not causal in general. The two results are compatible. The H$_2$ test is a half-filled two-site system, the setting in which Proposition~\ref{prop:ph} makes the closely related Hubbard-dimer splice pass the pointwise test for all couplings. Its non-causality there is visible only through the Pick test. Zhu and Chan~\cite{zhu21} anticipated that the total self-energy might not be strictly causal. Our criterion turns that caveat into a quantitative, checkable condition. The failure is distinct from the non-causal baths and impurity interactions of Refs.~\cite{lee17, nilsson17, chen22, backes22}: those studies concern causality problems in auxiliary quantities, such as the hybridization or the impurity interaction, rather than the assembled self-energy examined here. It is also distinct from the periodization problems of cluster DMFT~\cite{biroli02, biroli04, stanescu06, pelz26}, which come from truncating or interpolating a self-energy in momentum space. We note that non-causal behavior at the two lowest Matsubara frequencies has also been reported for cluster D-TRILEX self-energies~\cite{fossati26}. There it was attributed to missing diagrams; whether a criterion of the type~\eqref{eq:iff} applies is an open question.

\paragraph{Implications for production calculations.}
Proposition~\ref{thm:global} gives the condition: on a single-band lattice, the impurity spectral weight must satisfy Eq.~\eqref{eq:global} at every energy. On the Matsubara axis, where such calculations live, there are two levels of diagnostic. The pointwise test is cheapest: compare the impurity damping with $\gamma^{GW}-\min_{\mathbf k}\Gamma^{GW}(\mathbf k)$ (Proposition~\ref{prop:iff}), or, for several orbitals, compute the smallest eigenvalue of $\Gamma^{\GD}(\mathbf k,i\omega_n)$ over the Brillouin zone. It proves non-causality when it fails, but it is weak: in this survey it flags 54 points, whereas the eight-frequency Pick test flags 481. A block Pick test on a few low Matsubara frequencies, applied for each $\mathbf k$ or to the $\mathbf k$-resolved matrices, is far more sensitive. It is limited only by the precision of the data, because the Pick matrix becomes ill-conditioned as $N_P$ grows~\cite{fei21, fei21b}. Both complement the inspection of low-energy spectra used in Ref.~\cite{zhu21}. When either test fails, no causal self-energy, and hence no causal Green's function, reproduces the Matsubara data; an analytic continuation that matches the data cannot be causal. Where the resulting non-causality appears on the real axis depends on the continuation and is not fixed by the data. Possible remedies include a double counting chosen so that Proposition~\ref{prop:suff} holds, an explicitly causal projection of the data~\cite{han21}, or vertex corrections that make the local and nonlocal damping mutually consistent~\cite{toschi07, rohringer18, ayral15, stepanov19}, for which parquet implementations~\cite{valli15} and the high-frequency asymptotics of the vertex~\cite{wentzell20} are available, though the local vertex can itself become singular~\cite{schaefer13}. Schemes in which the nonlocal self-energy is static, as in the separability ansatz of Ref.~\cite{schaefer15}, avoid the problem altogether, since a static nonlocal part does not enter $\Gamma$.

\paragraph{Limitations.}
We use finite clusters, whose discrete spectra make the low-frequency damping depend on temperature through the thermal broadening of poles. Discreteness also makes the support mismatch of Corollary~\ref{cor:finite} unavoidable. The large number of Pick failures is therefore partly a finite-size effect, and a lattice calculation with continuous spectra may satisfy Eq.~\eqref{eq:global}. The pointwise failures, which require a large local overdamping at low energy, are the more transferable signal. $GW$ is evaluated once, on the exact $G$, rather than self-consistently, and the local self-energy is exact rather than a DMFT impurity self-energy. These choices isolate the splice. They do not show that any particular self-consistent $GW$+DMFT calculation is non-causal. The lattice forms of Propositions~\ref{thm:global} and~\ref{prop:iff} are exact, but we have not applied them here to a lattice $GW$+DMFT calculation. That is the obvious next step.

% ================================================================
\section{Conclusion}
\label{sec:conclusion}

Replacing the local part of a causal $GW$ self-energy by a causal local self-energy, even an exact one, does not guarantee a causal result. The spliced self-energy is causal if and only if its spectral measure, Eq.~\eqref{eq:BGD}, is positive semi-definite, a direct consequence of the Herglotz representation. Its more useful consequences are a momentum-resolved bound for translation-invariant systems, a generic failure mechanism on finite clusters, and Matsubara-axis diagnostics whose limits must be kept in view. In one-shot $GW$ on the exact cluster Green's function with the exact local self-energy inserted, the pointwise test fails at 54 of 528 parameter sets. These failures lie away from half filling, all at $\omega_1$, and appear already at weak coupling; each certifies non-causality. An eight-frequency block Pick test detects non-causality at 481 of the 528 sets, a lower bound for this survey. Moment checks do not detect any of these failures, because the spliced high-frequency moment equals the exact local one. These results concern an idealized splice, not a self-consistent $GW$+DMFT calculation, for which the conditions should be checked directly.

\paragraph{Note on the revision.}
Version~1 of this preprint reported causality violations of $GW$ itself, and of the spliced self-energy at half filling on the three-site ring. These came from a sign error in the polarization of the accompanying code, which produced an unscreened, Stoner-like denominator $1+U|P|\to1-U|P|$. With the sign fixed as described in Sec.~\ref{sec:method} and validated in App.~\ref{app:validation}, $GW$ passes the pointwise and Pick tests at every point, as the perfect-square theorem requires. The present version supersedes those results.

\section*{Data and code availability}
The exact-diagonalization and $GW$ code (\texttt{causality-check.py}), the scan and convergence scripts (\texttt{splice\_scan.py}, \texttt{converge\_splice.py}), the figure script (\texttt{make\_figures.py}), the second implementation and validation scripts (\texttt{indep\_check.py}, \texttt{indep\_pick.py}, \texttt{validation.py}), and the raw data (\texttt{splice\_scan.jsonl}, \texttt{indep\_pick.jsonl}, \texttt{validation.json}) are available at \url{https://doi.org/10.5281/zenodo.23644415}.

% ================================================================
\appendix

\section{Exact high-frequency moment}
\label{app:moments}

Inserting eigenstates $H|m\rangle=E_m|m\rangle$ with Boltzmann weights $p_m$ gives $G_{ij}(z)=\int d\varepsilon\,A_{ij}(\varepsilon)/(z-\varepsilon)$ with
$A_{ij}(\varepsilon)=\sum_{mn}(p_m+p_n)\langle m|c_i|n\rangle\langle n|c^\dagger_j|m\rangle\,\delta(\varepsilon-E_n+E_m)$, which is positive semi-definite. Its moments are anticommutators,
$M_0=I$, $M_1=\langle\{[c_i,H],c^\dagger_j\}\rangle$, $M_2=\langle\{[[c_i,H],H],c^\dagger_j\}\rangle$.
Inverting $G(z)=z^{-1}(I+M_1z^{-1}+M_2z^{-2}+\dots)$ gives $\Sigma_\infty=M_1-h$ and $C=M_2-M_1^2$.
With $[c_{i\sigma},H_U]=Ud_{i\bar\sigma}c_{i\sigma}$, one finds $M_1=h+U\,\mathrm{diag}(\langle d_{i\bar\sigma}\rangle)$. The cross terms in $M_2$ combine to $Uh_{ij}(\langle d_i\rangle+\langle d_j\rangle)$, and $d^2=\tfrac14$ gives $(U^2/4)\delta_{ij}$; the remaining term $\langle[d_{i\bar\sigma},H_t]\rangle$ vanishes by cyclicity of the trace. Subtracting $M_1^2$, the hopping and cross terms cancel identically, leaving
$C_{ij}=U^2(\tfrac14-\langle d_{i\bar\sigma}\rangle^2)\delta_{ij}=U^2\langle n_{i\bar\sigma}\rangle(1-\langle n_{i\bar\sigma}\rangle)\delta_{ij}$, i.e.\ Eq.~\eqref{eq:cex}~\cite{harris67, potthoff97}.
Numerically, $\omega_n\Gamma^{\ex}(i\omega_n)$ at $\omega_n\approx125$-$250$ reproduces Eq.~\eqref{eq:cex} site by site to $1.1\times10^{-3}$, the size of the $O(\omega_n^{-2})$ correction at that frequency.

\section{Causality and tail of $GW$}
\label{app:gw}

In real time the correlation part of Eq.~\eqref{eq:sgw} is a product $G_{ij}(t)\,[W-U]_{ij}(t)$. Its spectral function is
\begin{equation*}
\begin{split}B^{GW}_{ij}(\omega)=\!\int\! d\varepsilon\,d\varepsilon'\,&A_{ij}(\varepsilon)\,\rho_{ij}(\varepsilon')\,\big[1-f(\varepsilon)+b(\varepsilon')\big]\\ &\times\delta(\omega-\varepsilon-\varepsilon'),\end{split}
\end{equation*}
with $\rho$ the spectral function of $W-U$~\cite{strinati88}. The matrix $\rho(\varepsilon')[1-f(\varepsilon)+b(\varepsilon')]$ is positive semi-definite for both signs of $\varepsilon'$, because $\rho(\varepsilon')$ carries the sign of $\varepsilon'$ for a causal bosonic propagator~\cite{uimonen15}. By the Schur product theorem~\cite{horn13}, the element-wise product of two positive semi-definite matrices is positive semi-definite, so $B^{GW}(\omega)\succeq0$. This is the cluster version of the perfect-square argument of Ref.~\cite{stefanucci14}.
For the tail, $G_{ij}(i\omega_n+i\nu_m)=\delta_{ij}/(i\omega_n)+O(\omega_n^{-2})$ uniformly for $|\nu_m|\ll\omega_n$. Contributions from $|\nu_m|\sim\omega_n$ are suppressed by $[W-U](i\nu_m)=O(\nu_m^{-2})$. Together these give Eq.~\eqref{eq:cgw} with vanishing off-diagonal moment.

\section{Numerical details}
\label{app:numerics}

Green's functions are evaluated from the Lehmann representation, with transitions of weight below $10^{-13}$ discarded. They agree with the unpruned sum to $3\times10^{-13}$. Matsubara sums use fermionic and bosonic cutoffs $N_c=16\beta f$ and $M_c=8\beta f$, with $f=1,2$.

\paragraph{Tail correction of the polarization.}
Truncating Eq.~\eqref{eq:P} at $|n|\le N_c$ leaves a real, nearly frequency-independent error of order $\beta/N_c$ in $P(i\nu_m)$, because $G_{ij}(i\omega_n)\to\delta_{ij}/i\omega_n$. The error does not decay with $\nu_m$, so $W-U$ acquires a spurious constant and the high-frequency tail of $\Sigma^{GW}$ is wrong (it even turns negative). We therefore subtract $\delta_{ij}/[(i\omega_n+i\nu_m)i\omega_n]$ inside the window and add back its exact infinite sum: zero for $\nu_m\ne0$, and $-\beta/4$ per spin for $\nu_m=0$. With this correction, all violations change by less than $1.1\times10^{-6}$ between $f=1$ and $f=2$. Without it, the low-frequency violations converge only as $1/N_c$, to the same limits.

\paragraph{Checks.}
With the correction switched off, our $GW$ routine reproduces the original code to $2\times10^{-16}$. The causality test uses the 50 lowest positive Matsubara frequencies, and the moments use $\omega_n=(2n-1)\pi/\beta$ with $n=401$ and $801$. A full scan of the four clusters takes a few minutes on a laptop.

\section{Validation}
\label{app:validation}

\paragraph{Second implementation.}
The second implementation (\texttt{indep\_check.py}) shares no code with the main one. It builds the Fock-space Hamiltonian with its own Jordan-Wigner construction and an interleaved mode order, $(i,\sigma)\to2i+\sigma$. It evaluates $G_{ij}(\tau)=-Z^{-1}\mathrm{Tr}[e^{-(\beta-\tau)H}c_ie^{-\tau H}c^\dagger_j]$ on a composite Gauss-Legendre grid ($20\beta$ panels of 16 nodes), and forms $P_{ij}(\tau)=-2G_{ij}(\tau)G_{ji}(\beta-\tau)$ directly in imaginary time, so the bubble involves no fermionic Matsubara sum. $P(i\nu_m)$ follows by quadrature, $W-U$ is evaluated for $|m|\le300$, and $[W-U](\tau)$ is reconstructed with the analytic $1/\nu^2$ tail. The self-energy is $\Sigma^{GW}(\tau)=-G(\tau)\circ[W-U](\tau)$, Fourier transformed by quadrature, and $\Sigma^{\ex}=i\omega_n-h-G^{-1}$ uses $G(i\omega_n)$ from the same quadrature. Doubling the number of panels and of bosonic frequencies changes the damping eigenvalues by less than $2\times10^{-12}$.

\begin{table}[htbp]
\caption{Validation of the corrected polarization and screened interaction. $\chi$ is the exact density response at $U=0$ (dimer, open 3-chain, 3-site ring; $\beta t=10$; $\nu_0$-$\nu_5$). The screening checks cover all 528 points at $\nu_m$, $m\in\{0,1,2,4,8\}$.}
\label{tab:validation}
\centering
\footnotesize
\setlength{\tabcolsep}{4pt}
\begin{tabular}{@{}lr@{}}
\toprule
Check & result \\
\midrule
$\max|P-\chi|$, tail-corrected bubble & $1.1\times10^{-6}$ \\
$\max|P-\chi|$, uncorrected bubble ($N_c=16\beta$) & $6.3\times10^{-3}$ \\
$\max|P-\chi|$, Version~1 sign & $1.0$ \\
$\max\lambda[P(i\nu_m)]$ & $-2.1\times10^{-5}$ \\
$\min\lambda[1-UP(i\nu_m)]$ & $1.00001$ \\
$\max\lambda[W(0)]/U$ & $0.9997$ \\
$\min\lambda[1+UP(0)]$, Version~1 sign & $-10.4$ \\
$\lambda_{\min}[1+UP(0)]\le0$, Version~1 sign & 410/528 pts \\
\bottomrule
\end{tabular}
\end{table}

\begin{table}[htbp]
\caption{Smallest eigenvalue of $\Gamma^{\GD}(i\omega_1)$ from the main and the second implementation. The first five points fail the pointwise test; the last four pass it. The 3-site ring at $\beta t=10$, $U/t=4$, $\mu=0$ is the headline case of Version~1.}
\label{tab:indep}
\centering
\footnotesize
\setlength{\tabcolsep}{3pt}
\begin{tabular}{@{}lrrrrr@{}}
\toprule
cluster & $\beta t$ & $U/t$ & $\mu/t$ & main & second \\
\midrule
dimer  & 20 & 3 & $2.5$  & $-4.31709\times10^{-2}$ & $-4.31710\times10^{-2}$ \\
3-chain& 10 & 2 & $2.5$  & $-1.59693\times10^{-2}$ & $-1.59693\times10^{-2}$ \\
3-ring & 20 & 4 & $-3$   & $-2.41666\times10^{-2}$ & $-2.41668\times10^{-2}$ \\
4-ring & 20 & 4 & $1.5$  & $-1.54666\times10^{-2}$ & $-1.54667\times10^{-2}$ \\
4-ring & 20 & 3 & $-3$   & $-1.586\times10^{-5}$   & $-1.590\times10^{-5}$ \\
dimer  & 20 & 3 & $1$    & $5.31925\times10^{-2}$  & $5.31925\times10^{-2}$ \\
3-chain& 10 & 2 & $0$    & $2.32459\times10^{-2}$  & $2.32459\times10^{-2}$ \\
3-ring & 10 & 4 & $0$    & $8.57144\times10^{-1}$  & $8.57144\times10^{-1}$ \\
4-ring & 10 & 2 & $-3$   & $2.1048\times10^{-4}$   & $2.1045\times10^{-4}$ \\
\bottomrule
\end{tabular}
\end{table}

\paragraph{Pick test.}
The block Pick matrix~\eqref{eq:pick} is scaled by its diagonal, $\tilde{\mathcal P}=D^{-1/2}\mathcal PD^{-1/2}$ with $D$ the diagonal of $\mathcal P$, which leaves the sign of its eigenvalues unchanged. With the second implementation, the smallest scaled eigenvalue over all points and $N_P\le8$ is $-4\times10^{-12}$ for $\Sigma^{\ex}$ and $-1\times10^{-15}$ for $\Sigma^{GW}$; we count a failure below $-10^{-9}$. With the main implementation the controls show noise up to $2\times10^{-6}$ for $N_P\ge3$, because the Pick matrix is ill-conditioned and $\Sigma^{GW}$ is accurate only to $\sim10^{-6}$. Its $N_P=2$ counts agree with those of Table~\ref{tab:summary}.

\bibliographystyle{unsrtnat}
\bibliography{refs}

\begin{thebibliography}{84}
\providecommand{\natexlab}[1]{#1}
\providecommand{\url}[1]{\texttt{#1}}
\expandafter\ifx\csname urlstyle\endcsname\relax
  \providecommand{\doi}[1]{doi: #1}\else
  \providecommand{\doi}{doi: \begingroup \urlstyle{rm}\Url}\fi

\bibitem[Hedin(1965)]{hedin65}
L.~Hedin.
\newblock New method for calculating the one-particle {G}reen's function with application to the electron-gas problem.
\newblock \emph{Phys. Rev.}, 139:\penalty0 A796, 1965.
\newblock \doi{10.1103/PhysRev.139.A796}.

\bibitem[Strinati et~al.(1980)Strinati, Mattausch, and Hanke]{strinati80}
G.~Strinati, H.~J. Mattausch, and W.~Hanke.
\newblock Dynamical correlation effects on the quasiparticle {B}loch states of a covalent crystal.
\newblock \emph{Phys. Rev. Lett.}, 45:\penalty0 290, 1980.
\newblock \doi{10.1103/PhysRevLett.45.290}.

\bibitem[Strinati et~al.(1982)Strinati, Mattausch, and Hanke]{strinati82}
G.~Strinati, H.~J. Mattausch, and W.~Hanke.
\newblock Dynamical aspects of correlation corrections in a covalent crystal.
\newblock \emph{Phys. Rev. B}, 25:\penalty0 2867, 1982.
\newblock \doi{10.1103/PhysRevB.25.2867}.

\bibitem[Hybertsen and Louie(1986)]{hybertsen86}
M.~S. Hybertsen and S.~G. Louie.
\newblock Electron correlation in semiconductors and insulators: Band gaps and quasiparticle energies.
\newblock \emph{Phys. Rev. B}, 34:\penalty0 5390, 1986.
\newblock \doi{10.1103/PhysRevB.34.5390}.

\bibitem[Strinati(1988)]{strinati88}
G.~Strinati.
\newblock Application of the {G}reen's functions method to the study of the optical properties of semiconductors.
\newblock \emph{Riv. Nuovo Cimento}, 11:\penalty0 1, 1988.
\newblock \doi{10.1007/BF02725962}.

\bibitem[Aryasetiawan and Gunnarsson(1998)]{aryasetiawan98}
F.~Aryasetiawan and O.~Gunnarsson.
\newblock The {GW} method.
\newblock \emph{Rep. Prog. Phys.}, 61:\penalty0 237, 1998.

\bibitem[Onida et~al.(2002)Onida, Reining, and Rubio]{onida02}
G.~Onida, L.~Reining, and A.~Rubio.
\newblock Electronic excitations: density-functional versus many-body {G}reen's-function approaches.
\newblock \emph{Rev. Mod. Phys.}, 74:\penalty0 601, 2002.
\newblock \doi{10.1103/RevModPhys.74.601}.

\bibitem[Reining(2018)]{reining18}
L.~Reining.
\newblock The {GW} approximation: content, successes and limitations.
\newblock \emph{WIREs Comput. Mol. Sci.}, 8:\penalty0 e1344, 2018.

\bibitem[Golze et~al.(2019)Golze, Dvorak, and Rinke]{golze19}
D.~Golze, M.~Dvorak, and P.~Rinke.
\newblock The {GW} compendium: A practical guide to theoretical photoemission spectroscopy.
\newblock \emph{Front. Chem.}, 7:\penalty0 377, 2019.

\bibitem[Metzner and Vollhardt(1989)]{metzner89}
W.~Metzner and D.~Vollhardt.
\newblock Correlated lattice fermions in $d=\infty$ dimensions.
\newblock \emph{Phys. Rev. Lett.}, 62:\penalty0 324, 1989.
\newblock \doi{10.1103/PhysRevLett.62.324}.

\bibitem[Georges and Kotliar(1992)]{georges92}
A.~Georges and G.~Kotliar.
\newblock {H}ubbard model in infinite dimensions.
\newblock \emph{Phys. Rev. B}, 45:\penalty0 6479, 1992.
\newblock \doi{10.1103/PhysRevB.45.6479}.

\bibitem[Georges et~al.(1996)Georges, Kotliar, Krauth, and Rozenberg]{georges96}
A.~Georges, G.~Kotliar, W.~Krauth, and M.~J. Rozenberg.
\newblock Dynamical mean-field theory of strongly correlated fermion systems and the limit of infinite dimensions.
\newblock \emph{Rev. Mod. Phys.}, 68:\penalty0 13, 1996.
\newblock \doi{10.1103/RevModPhys.68.13}.

\bibitem[Kotliar et~al.(2006)Kotliar, Savrasov, Haule, Oudovenko, Parcollet, and Marianetti]{kotliar06}
G.~Kotliar, S.~Y. Savrasov, K.~Haule, V.~S. Oudovenko, O.~Parcollet, and C.~A. Marianetti.
\newblock Electronic structure calculations with dynamical mean-field theory.
\newblock \emph{Rev. Mod. Phys.}, 78:\penalty0 865, 2006.
\newblock \doi{10.1103/RevModPhys.78.865}.

\bibitem[Held(2007)]{held07}
K.~Held.
\newblock Electronic structure calculations using dynamical mean field theory.
\newblock \emph{Adv. Phys.}, 56:\penalty0 829, 2007.

\bibitem[Strinati(1984)]{strinati84}
G.~Strinati.
\newblock Effects of dynamical screening on resonances at inner-shell thresholds in semiconductors.
\newblock \emph{Phys. Rev. B}, 29:\penalty0 5718, 1984.
\newblock \doi{10.1103/PhysRevB.29.5718}.

\bibitem[Aryasetiawan et~al.(2004)Aryasetiawan, Imada, Georges, Kotliar, Biermann, and Lichtenstein]{aryasetiawan04}
F.~Aryasetiawan, M.~Imada, A.~Georges, G.~Kotliar, S.~Biermann, and A.~I. Lichtenstein.
\newblock Frequency-dependent local interactions and low-energy effective models from electronic structure calculations.
\newblock \emph{Phys. Rev. B}, 70:\penalty0 195104, 2004.
\newblock \doi{10.1103/PhysRevB.70.195104}.

\bibitem[Sun and Kotliar(2002)]{sun02}
P.~Sun and G.~Kotliar.
\newblock Extended dynamical mean-field theory and {GW} method.
\newblock \emph{Phys. Rev. B}, 66:\penalty0 085120, 2002.
\newblock \doi{10.1103/PhysRevB.66.085120}.

\bibitem[Biermann et~al.(2003)Biermann, Aryasetiawan, and Georges]{biermann03}
S.~Biermann, F.~Aryasetiawan, and A.~Georges.
\newblock First-principles approach to the electronic structure of strongly correlated systems: Combining the {GW} approximation and dynamical mean-field theory.
\newblock \emph{Phys. Rev. Lett.}, 90:\penalty0 086402, 2003.
\newblock \doi{10.1103/PhysRevLett.90.086402}.

\bibitem[Sun and Kotliar(2004)]{sun04}
P.~Sun and G.~Kotliar.
\newblock Many-body approximation scheme beyond {GW}.
\newblock \emph{Phys. Rev. Lett.}, 92:\penalty0 196402, 2004.
\newblock \doi{10.1103/PhysRevLett.92.196402}.

\bibitem[Ayral et~al.(2012)Ayral, Werner, and Biermann]{ayral12}
T.~Ayral, P.~Werner, and S.~Biermann.
\newblock Spectral properties of correlated materials: Local vertex and nonlocal two-particle correlations from combined {GW} and dynamical mean field theory.
\newblock \emph{Phys. Rev. Lett.}, 109:\penalty0 226401, 2012.
\newblock \doi{10.1103/PhysRevLett.109.226401}.

\bibitem[Ayral et~al.(2013)Ayral, Biermann, and Werner]{ayral13}
T.~Ayral, S.~Biermann, and P.~Werner.
\newblock Screening and nonlocal correlations in the extended {H}ubbard model from self-consistent combined {GW} and dynamical mean field theory.
\newblock \emph{Phys. Rev. B}, 87:\penalty0 125149, 2013.
\newblock \doi{10.1103/PhysRevB.87.125149}.

\bibitem[Tomczak et~al.(2012)Tomczak, Casula, Miyake, Aryasetiawan, and Biermann]{tomczak12}
J.~M. Tomczak, M.~Casula, T.~Miyake, F.~Aryasetiawan, and S.~Biermann.
\newblock Combined {GW} and dynamical mean-field theory: Dynamical screening effects in transition metal oxides.
\newblock \emph{EPL}, 100:\penalty0 67001, 2012.

\bibitem[Biermann(2014)]{biermann14}
S.~Biermann.
\newblock Dynamical screening effects in correlated electron materials-a progress report on combined many-body perturbation and dynamical mean field theory: `{GW}+{DMFT}'.
\newblock \emph{J. Phys.: Condens. Matter}, 26:\penalty0 173202, 2014.

\bibitem[Huang et~al.(2014)Huang, Ayral, Biermann, and Werner]{huang14}
L.~Huang, T.~Ayral, S.~Biermann, and P.~Werner.
\newblock Extended dynamical mean-field study of the {H}ubbard model with long-range interactions.
\newblock \emph{Phys. Rev. B}, 90:\penalty0 195114, 2014.
\newblock \doi{10.1103/PhysRevB.90.195114}.

\bibitem[van Roekeghem et~al.(2014)van Roekeghem, Ayral, Tomczak, Casula, Xu, Ding, Ferrero, Parcollet, Jiang, and Biermann]{vanroekeghem14}
A.~van Roekeghem, T.~Ayral, J.~M. Tomczak, M.~Casula, N.~Xu, H.~Ding, M.~Ferrero, O.~Parcollet, H.~Jiang, and S.~Biermann.
\newblock Dynamical correlations and screened exchange on the experimental bench: Spectral properties of the cobalt pnictide {BaCo$_2$As$_2$}.
\newblock \emph{Phys. Rev. Lett.}, 113:\penalty0 266403, 2014.
\newblock \doi{10.1103/PhysRevLett.113.266403}.

\bibitem[Tomczak(2015)]{tomczak15}
J.~M. Tomczak.
\newblock {QSGW+DMFT}: an electronic structure scheme for the iron pnictides and beyond, 2015.
\newblock arXiv:1411.5180.

\bibitem[Choi et~al.(2016)Choi, Kutepov, Haule, van Schilfgaarde, and Kotliar]{choi16}
S.~Choi, A.~Kutepov, K.~Haule, M.~van Schilfgaarde, and G.~Kotliar.
\newblock First-principles treatment of {M}ott insulators: linearized {QSGW+DMFT} approach.
\newblock \emph{npj Quantum Mater.}, 1:\penalty0 16001, 2016.

\bibitem[Boehnke et~al.(2016)Boehnke, Nilsson, Aryasetiawan, and Werner]{boehnke16}
L.~Boehnke, F.~Nilsson, F.~Aryasetiawan, and P.~Werner.
\newblock When strong correlations become weak: Consistent merging of {GW} and {DMFT}.
\newblock \emph{Phys. Rev. B}, 94:\penalty0 201106(R), 2016.
\newblock \doi{10.1103/PhysRevB.94.201106}.

\bibitem[Nilsson et~al.(2017)Nilsson, Boehnke, Werner, and Aryasetiawan]{nilsson17}
F.~Nilsson, L.~Boehnke, P.~Werner, and F.~Aryasetiawan.
\newblock Multitier self-consistent {GW}+{EDMFT}.
\newblock \emph{Phys. Rev. Mater.}, 1:\penalty0 043803, 2017.
\newblock \doi{10.1103/PhysRevMaterials.1.043803}.

\bibitem[Gole\v{z} et~al.(2017)Gole\v{z}, Boehnke, Strand, Eckstein, and Werner]{golez17}
D.~Gole\v{z}, L.~Boehnke, H.~U.~R. Strand, M.~Eckstein, and P.~Werner.
\newblock Nonequilibrium {GW}+{EDMFT}: Antiscreening and inverted populations from nonlocal correlations.
\newblock \emph{Phys. Rev. Lett.}, 118:\penalty0 246402, 2017.
\newblock \doi{10.1103/PhysRevLett.118.246402}.

\bibitem[Choi et~al.(2019)Choi, Semon, Kang, Kutepov, and Kotliar]{choi19}
S.~Choi, P.~Semon, B.~Kang, A.~Kutepov, and G.~Kotliar.
\newblock {ComDMFT}: A massively parallel computer package for the electronic structure of correlated-electron systems.
\newblock \emph{Comput. Phys. Commun.}, 244:\penalty0 277, 2019.

\bibitem[Zhu and Chan(2021)]{zhu21}
T.~Zhu and G.~K.-L. Chan.
\newblock Ab initio full cell {GW}+{DMFT} for correlated materials.
\newblock \emph{Phys. Rev. X}, 11:\penalty0 021006, 2021.
\newblock \doi{10.1103/PhysRevX.11.021006}.

\bibitem[Zgid and Gull(2017)]{zgid17}
D.~Zgid and E.~Gull.
\newblock Finite temperature quantum embedding theories for correlated systems.
\newblock \emph{New J. Phys.}, 19:\penalty0 023047, 2017.

\bibitem[Lan et~al.(2017)Lan, Shee, Li, Gull, and Zgid]{lan17}
T.~N. Lan, A.~Shee, J.~Li, E.~Gull, and D.~Zgid.
\newblock Testing self-energy embedding theory in combination with {GW}.
\newblock \emph{Phys. Rev. B}, 96:\penalty0 155106, 2017.
\newblock \doi{10.1103/PhysRevB.96.155106}.

\bibitem[Iskakov et~al.(2020)Iskakov, Yeh, Gull, and Zgid]{iskakov20}
S.~Iskakov, C.-N. Yeh, E.~Gull, and D.~Zgid.
\newblock Ab initio self-energy embedding for the photoemission spectra of {NiO} and {MnO}.
\newblock \emph{Phys. Rev. B}, 102:\penalty0 085105, 2020.
\newblock \doi{10.1103/PhysRevB.102.085105}.

\bibitem[Toschi et~al.(2007)Toschi, Katanin, and Held]{toschi07}
A.~Toschi, A.~A. Katanin, and K.~Held.
\newblock Dynamical vertex approximation: A step beyond dynamical mean-field theory.
\newblock \emph{Phys. Rev. B}, 75:\penalty0 045118, 2007.
\newblock \doi{10.1103/PhysRevB.75.045118}.

\bibitem[Katanin et~al.(2009)Katanin, Toschi, and Held]{katanin09}
A.~A. Katanin, A.~Toschi, and K.~Held.
\newblock Comparing pertinent effects of antiferromagnetic fluctuations in the two- and three-dimensional {H}ubbard model.
\newblock \emph{Phys. Rev. B}, 80:\penalty0 075104, 2009.
\newblock \doi{10.1103/PhysRevB.80.075104}.

\bibitem[Rubtsov et~al.(2008)Rubtsov, Katsnelson, and Lichtenstein]{rubtsov08}
A.~N. Rubtsov, M.~I. Katsnelson, and A.~I. Lichtenstein.
\newblock Dual fermion approach to nonlocal correlations in the {H}ubbard model.
\newblock \emph{Phys. Rev. B}, 77:\penalty0 033101, 2008.
\newblock \doi{10.1103/PhysRevB.77.033101}.

\bibitem[Rubtsov et~al.(2012)Rubtsov, Katsnelson, and Lichtenstein]{rubtsov12}
A.~N. Rubtsov, M.~I. Katsnelson, and A.~I. Lichtenstein.
\newblock Dual boson approach to collective excitations in correlated fermionic systems.
\newblock \emph{Ann. Phys.}, 327:\penalty0 1320, 2012.

\bibitem[Ayral and Parcollet(2015)]{ayral15}
T.~Ayral and O.~Parcollet.
\newblock Mott physics and spin fluctuations: A unified framework.
\newblock \emph{Phys. Rev. B}, 92:\penalty0 115109, 2015.
\newblock \doi{10.1103/PhysRevB.92.115109}.

\bibitem[Stepanov et~al.(2019)Stepanov, Harkov, and Lichtenstein]{stepanov19}
E.~A. Stepanov, V.~Harkov, and A.~I. Lichtenstein.
\newblock Consistent partial bosonization of the extended {H}ubbard model.
\newblock \emph{Phys. Rev. B}, 100:\penalty0 205115, 2019.
\newblock \doi{10.1103/PhysRevB.100.205115}.

\bibitem[Rohringer et~al.(2018)Rohringer, Hafermann, Toschi, Katanin, Antipov, Katsnelson, Lichtenstein, Rubtsov, and Held]{rohringer18}
G.~Rohringer, H.~Hafermann, A.~Toschi, A.~A. Katanin, A.~E. Antipov, M.~I. Katsnelson, A.~I. Lichtenstein, A.~N. Rubtsov, and K.~Held.
\newblock Diagrammatic routes to nonlocal correlations beyond dynamical mean field theory.
\newblock \emph{Rev. Mod. Phys.}, 90:\penalty0 025003, 2018.
\newblock \doi{10.1103/RevModPhys.90.025003}.

\bibitem[Luttinger(1961)]{luttinger61}
J.~M. Luttinger.
\newblock Analytic properties of single-particle propagators for many-fermion systems.
\newblock \emph{Phys. Rev.}, 121:\penalty0 942, 1961.
\newblock \doi{10.1103/PhysRev.121.942}.

\bibitem[Gesztesy and Tsekanovskii(2000)]{gesztesy00}
F.~Gesztesy and E.~Tsekanovskii.
\newblock On matrix-valued {H}erglotz functions.
\newblock \emph{Math. Nachr.}, 218:\penalty0 61, 2000.

\bibitem[Stefanucci et~al.(2014)Stefanucci, Pavlyukh, Uimonen, and van Leeuwen]{stefanucci14}
G.~Stefanucci, Y.~Pavlyukh, A.-M. Uimonen, and R.~van Leeuwen.
\newblock Diagrammatic expansion for positive spectral functions beyond {GW}: Application to vertex corrections in the electron gas.
\newblock \emph{Phys. Rev. B}, 90:\penalty0 115134, 2014.
\newblock \doi{10.1103/PhysRevB.90.115134}.

\bibitem[Uimonen et~al.(2015)Uimonen, Stefanucci, Pavlyukh, and van Leeuwen]{uimonen15}
A.-M. Uimonen, G.~Stefanucci, Y.~Pavlyukh, and R.~van Leeuwen.
\newblock Diagrammatic expansion for positive density-response spectra: Application to the electron gas, 2015.
\newblock arXiv:1503.02158.

\bibitem[Baym and Kadanoff(1961)]{baym61}
G.~Baym and L.~P. Kadanoff.
\newblock Conservation laws and correlation functions.
\newblock \emph{Phys. Rev.}, 124:\penalty0 287, 1961.
\newblock \doi{10.1103/PhysRev.124.287}.

\bibitem[Baym(1962)]{baym62}
G.~Baym.
\newblock Self-consistent approximations in many-body systems.
\newblock \emph{Phys. Rev.}, 127:\penalty0 1391, 1962.
\newblock \doi{10.1103/PhysRev.127.1391}.

\bibitem[Almbladh et~al.(1999)Almbladh, von Barth, and van Leeuwen]{almbladh99}
C.-O. Almbladh, U.~von Barth, and R.~van Leeuwen.
\newblock Variational total energies from {$\Phi$}- and {$\Psi$}-derivable theories.
\newblock \emph{Int. J. Mod. Phys. B}, 13:\penalty0 535, 1999.

\bibitem[Lee and Haule(2017)]{lee17}
J.~Lee and K.~Haule.
\newblock Diatomic molecule as a testbed for combining {DMFT} with electronic structure methods such as {GW} and {DFT}.
\newblock \emph{Phys. Rev. B}, 95:\penalty0 155104, 2017.
\newblock \doi{10.1103/PhysRevB.95.155104}.

\bibitem[Lee(2017)]{leethesis17}
J.~Lee.
\newblock \emph{Methodological study of computational approaches to address the problem of strong correlations}.
\newblock PhD thesis, Rutgers, The State University of New Jersey, 2017.

\bibitem[Chen et~al.(2022)Chen, Petocchi, and Werner]{chen22}
J.~Chen, F.~Petocchi, and P.~Werner.
\newblock Causal versus local {GW}+{EDMFT} scheme and application to the triangular-lattice extended {H}ubbard model.
\newblock \emph{Phys. Rev. B}, 105:\penalty0 085102, 2022.
\newblock \doi{10.1103/PhysRevB.105.085102}.

\bibitem[Backes et~al.(2022)Backes, Sim, and Biermann]{backes22}
S.~Backes, J.-H. Sim, and S.~Biermann.
\newblock Nonlocal correlation effects in fermionic many-body systems: Overcoming the noncausality problem.
\newblock \emph{Phys. Rev. B}, 105:\penalty0 245115, 2022.
\newblock \doi{10.1103/PhysRevB.105.245115}.

\bibitem[Hettler et~al.(1998)Hettler, Tahvildar-Zadeh, Jarrell, Pruschke, and Krishnamurthy]{hettler98}
M.~H. Hettler, A.~N. Tahvildar-Zadeh, M.~Jarrell, T.~Pruschke, and H.~R. Krishnamurthy.
\newblock Nonlocal dynamical correlations of strongly interacting electron systems.
\newblock \emph{Phys. Rev. B}, 58:\penalty0 R7475, 1998.
\newblock \doi{10.1103/PhysRevB.58.R7475}.

\bibitem[Lichtenstein and Katsnelson(2000)]{lichtenstein00}
A.~I. Lichtenstein and M.~I. Katsnelson.
\newblock Antiferromagnetism and d-wave superconductivity in cuprates: A cluster dynamical mean-field theory.
\newblock \emph{Phys. Rev. B}, 62:\penalty0 R9283, 2000.
\newblock \doi{10.1103/PhysRevB.62.R9283}.

\bibitem[Kotliar et~al.(2001)Kotliar, Savrasov, P\'alsson, and Biroli]{kotliar01}
G.~Kotliar, S.~Y. Savrasov, G.~P\'alsson, and G.~Biroli.
\newblock Cellular approach to the {H}ubbard model.
\newblock \emph{Phys. Rev. Lett.}, 87:\penalty0 186401, 2001.
\newblock \doi{10.1103/PhysRevLett.87.186401}.

\bibitem[Hettler et~al.(2000)Hettler, Mukherjee, Jarrell, and Krishnamurthy]{hettler00}
M.~H. Hettler, M.~Mukherjee, M.~Jarrell, and H.~R. Krishnamurthy.
\newblock Dynamical cluster approximation: Nonlocal dynamics of correlated electron systems.
\newblock \emph{Phys. Rev. B}, 61:\penalty0 12739, 2000.
\newblock \doi{10.1103/PhysRevB.61.12739}.

\bibitem[Biroli and Kotliar(2002)]{biroli02}
G.~Biroli and G.~Kotliar.
\newblock Cluster methods for strongly correlated electron systems.
\newblock \emph{Phys. Rev. B}, 65:\penalty0 155112, 2002.
\newblock \doi{10.1103/PhysRevB.65.155112}.

\bibitem[Biroli et~al.(2004)Biroli, Parcollet, and Kotliar]{biroli04}
G.~Biroli, O.~Parcollet, and G.~Kotliar.
\newblock Cluster dynamical mean-field theories: Causality and classical limit.
\newblock \emph{Phys. Rev. B}, 69:\penalty0 205108, 2004.
\newblock \doi{10.1103/PhysRevB.69.205108}.

\bibitem[Maier et~al.(2005)Maier, Jarrell, Pruschke, and Hettler]{maier05}
T.~Maier, M.~Jarrell, T.~Pruschke, and M.~H. Hettler.
\newblock Quantum cluster theories.
\newblock \emph{Rev. Mod. Phys.}, 77:\penalty0 1027, 2005.
\newblock \doi{10.1103/RevModPhys.77.1027}.

\bibitem[Stanescu and Kotliar(2006)]{stanescu06}
T.~D. Stanescu and G.~Kotliar.
\newblock Fermi arcs and hidden zeros of the {G}reen function in the pseudogap state.
\newblock \emph{Phys. Rev. B}, 74:\penalty0 125110, 2006.
\newblock \doi{10.1103/PhysRevB.74.125110}.

\bibitem[Pelz et~al.(2026)Pelz, von Delft, and Gleis]{pelz26}
M.~Pelz, J.~von Delft, and A.~Gleis.
\newblock Liouvillian interpolation of the self-energy of cluster dynamical mean-field theories, 2026.
\newblock arXiv:2602.16351.

\bibitem[Hubbard(1963)]{hubbard63}
J.~Hubbard.
\newblock Electron correlations in narrow energy bands.
\newblock \emph{Proc. R. Soc. Lond. A}, 276:\penalty0 238, 1963.

\bibitem[Carrascal et~al.(2015)Carrascal, Ferrer, Smith, and Burke]{carrascal15}
D.~J. Carrascal, J.~Ferrer, J.~C. Smith, and K.~Burke.
\newblock The {H}ubbard dimer: a density functional case study of a many-body problem.
\newblock \emph{J. Phys.: Condens. Matter}, 27:\penalty0 393001, 2015.

\bibitem[Romaniello et~al.(2009)Romaniello, Guyot, and Reining]{romaniello09}
P.~Romaniello, S.~Guyot, and L.~Reining.
\newblock The self-energy beyond {GW}: Local and nonlocal vertex corrections.
\newblock \emph{J. Chem. Phys.}, 131:\penalty0 154111, 2009.

\bibitem[Romaniello et~al.(2012)Romaniello, Bechstedt, and Reining]{romaniello12}
P.~Romaniello, F.~Bechstedt, and L.~Reining.
\newblock Beyond the {GW} approximation: Combining correlation channels.
\newblock \emph{Phys. Rev. B}, 85:\penalty0 155131, 2012.
\newblock \doi{10.1103/PhysRevB.85.155131}.

\bibitem[Holm and von Barth(1998)]{holm98}
B.~Holm and U.~von Barth.
\newblock Fully self-consistent {GW} self-energy of the electron gas.
\newblock \emph{Phys. Rev. B}, 57:\penalty0 2108, 1998.
\newblock \doi{10.1103/PhysRevB.57.2108}.

\bibitem[van Schilfgaarde et~al.(2006)van Schilfgaarde, Kotani, and Faleev]{vanschilfgaarde06}
M.~van Schilfgaarde, T.~Kotani, and S.~Faleev.
\newblock Quasiparticle self-consistent {GW} theory.
\newblock \emph{Phys. Rev. Lett.}, 96:\penalty0 226402, 2006.
\newblock \doi{10.1103/PhysRevLett.96.226402}.

\bibitem[Shishkin et~al.(2007)Shishkin, Marsman, and Kresse]{shishkin07}
M.~Shishkin, M.~Marsman, and G.~Kresse.
\newblock Accurate quasiparticle spectra from self-consistent {GW} calculations with vertex corrections.
\newblock \emph{Phys. Rev. Lett.}, 99:\penalty0 246403, 2007.
\newblock \doi{10.1103/PhysRevLett.99.246403}.

\bibitem[Stefanucci and van Leeuwen(2013)]{stefanucci13}
G.~Stefanucci and R.~van Leeuwen.
\newblock \emph{Nonequilibrium Many-Body Theory of Quantum Systems: A Modern Introduction}.
\newblock Cambridge University Press, Cambridge, 2013.

\bibitem[Srivastava(2021)]{srivastava21}
A.~Srivastava.
\newblock Causality and passivity: From electromagnetism and network theory to metamaterials.
\newblock \emph{Mech. Mater.}, 154:\penalty0 103710, 2021.
\newblock arXiv:2008.05546.

\bibitem[Fei et~al.(2021{\natexlab{a}})Fei, Yeh, and Gull]{fei21}
J.~Fei, C.-N. Yeh, and E.~Gull.
\newblock Nevanlinna analytical continuation.
\newblock \emph{Phys. Rev. Lett.}, 126:\penalty0 056402, 2021{\natexlab{a}}.
\newblock \doi{10.1103/PhysRevLett.126.056402}.

\bibitem[Fei et~al.(2021{\natexlab{b}})Fei, Yeh, Zgid, and Gull]{fei21b}
J.~Fei, C.-N. Yeh, D.~Zgid, and E.~Gull.
\newblock Analytical continuation of matrix-valued functions: {C}arath\'eodory formalism.
\newblock \emph{Phys. Rev. B}, 104:\penalty0 165111, 2021{\natexlab{b}}.
\newblock \doi{10.1103/PhysRevB.104.165111}.

\bibitem[Vidberg and Serene(1977)]{vidberg77}
H.~J. Vidberg and J.~W. Serene.
\newblock Solving the {E}liashberg equations by means of {N}-point {P}ad\'e approximants.
\newblock \emph{J. Low Temp. Phys.}, 29:\penalty0 179, 1977.

\bibitem[Jarrell and Gubernatis(1996)]{jarrell96}
M.~Jarrell and J.~E. Gubernatis.
\newblock Bayesian inference and the analytic continuation of imaginary-time quantum {M}onte {C}arlo data.
\newblock \emph{Phys. Rep.}, 269:\penalty0 133, 1996.

\bibitem[Harris and Lange(1967)]{harris67}
A.~B. Harris and R.~V. Lange.
\newblock Single-particle excitations in narrow energy bands.
\newblock \emph{Phys. Rev.}, 157:\penalty0 295, 1967.
\newblock \doi{10.1103/PhysRev.157.295}.

\bibitem[Potthoff et~al.(1997)Potthoff, Wegner, and Nolting]{potthoff97}
M.~Potthoff, T.~Wegner, and W.~Nolting.
\newblock Interpolating self-energy of the infinite-dimensional {H}ubbard model: Modifying the iterative perturbation theory.
\newblock \emph{Phys. Rev. B}, 55:\penalty0 16132, 1997.
\newblock \doi{10.1103/PhysRevB.55.16132}.

\bibitem[Fossati and Stepanov(2026)]{fossati26}
F.~Fossati and E.~A. Stepanov.
\newblock Dual-space cluster-diagrammatic approach to nonlocal electronic correlations, 2026.
\newblock Phys. Rev. B (2026), doi:10.1103/dhm8-5ss6; arXiv:2507.06015.

\bibitem[Han and Choi(2021)]{han21}
M.~Han and H.~J. Choi.
\newblock Causal optimization method for imaginary-time {G}reen's functions in interacting electron systems.
\newblock \emph{Phys. Rev. B}, 104:\penalty0 115112, 2021.
\newblock \doi{10.1103/PhysRevB.104.115112}.

\bibitem[Valli et~al.(2015)Valli, Sch\"afer, Thunstr\"om, Rohringer, Andergassen, Sangiovanni, Held, and Toschi]{valli15}
A.~Valli, T.~Sch\"afer, P.~Thunstr\"om, G.~Rohringer, S.~Andergassen, G.~Sangiovanni, K.~Held, and A.~Toschi.
\newblock Dynamical vertex approximation in its parquet implementation: Application to {H}ubbard nanorings.
\newblock \emph{Phys. Rev. B}, 91:\penalty0 115115, 2015.
\newblock \doi{10.1103/PhysRevB.91.115115}.

\bibitem[Wentzell et~al.(2020)Wentzell, Li, Tagliavini, Taranto, Rohringer, Held, Toschi, and Andergassen]{wentzell20}
N.~Wentzell, G.~Li, A.~Tagliavini, C.~Taranto, G.~Rohringer, K.~Held, A.~Toschi, and S.~Andergassen.
\newblock High-frequency asymptotics of the vertex function: Diagrammatic parametrization and algorithmic implementation.
\newblock \emph{Phys. Rev. B}, 102:\penalty0 085106, 2020.
\newblock \doi{10.1103/PhysRevB.102.085106}.

\bibitem[Sch\"afer et~al.(2013)Sch\"afer, Rohringer, Gunnarsson, Ciuchi, Sangiovanni, and Toschi]{schaefer13}
T.~Sch\"afer, G.~Rohringer, O.~Gunnarsson, S.~Ciuchi, G.~Sangiovanni, and A.~Toschi.
\newblock Divergent precursors of the {M}ott-{H}ubbard transition at the two-particle level.
\newblock \emph{Phys. Rev. Lett.}, 110:\penalty0 246405, 2013.
\newblock \doi{10.1103/PhysRevLett.110.246405}.

\bibitem[Sch\"afer et~al.(2015)Sch\"afer, Toschi, and Tomczak]{schaefer15}
T.~Sch\"afer, A.~Toschi, and J.~M. Tomczak.
\newblock Separability of dynamical and nonlocal correlations in three dimensions.
\newblock \emph{Phys. Rev. B}, 91:\penalty0 121107(R), 2015.
\newblock \doi{10.1103/PhysRevB.91.121107}.

\bibitem[Horn and Johnson(2013)]{horn13}
R.~A. Horn and C.~R. Johnson.
\newblock \emph{Matrix Analysis}.
\newblock Cambridge University Press, Cambridge, 2nd edition, 2013.

\end{thebibliography}

\end{document}